\documentclass{iacrtrans}

\usepackage[utf8]{inputenc}
\usepackage{etex}
\usepackage{microtype}

\usepackage{amsmath,amssymb}
\usepackage{algorithm,algorithmicx}
\usepackage{algpseudocode}
\usepackage{booktabs}
\usepackage{mathtools}
\usepackage{xspace}
\usepackage{xcolor}
\usepackage{url}
\usepackage{enumitem}

\let\|\relax
\DeclareMathSymbol{\|}{\mathbin}{symbols}{"6B}

\newcommand{\SSS}{\textup S}

\newcommand{\CC}{\textup C}
\newcommand{\DD}{\mathcal D}
\newcommand{\mathin}{\mathrm{in}}
\newcommand{\mathout}{\mathrm{out}}

\newcommand{\mathfin}{\mathrm{fin}}

\newcommand{\ket}[1]{| #1 \rangle}

\newcommand{\Q}[1]{\text{Q#1}}

\begin{document}

\author{Marc Kaplan\inst{1,2} \and Ga\"etan Leurent\inst{3} \and
Anthony Leverrier\inst{3} \and Mar\'ia Naya-Plasencia\inst{3}}

\institute{LTCI, T\'el\'ecom ParisTech, 23 avenue d'Italie, 75214 Paris CEDEX 13, France
\and
School of Informatics, University of Edinburgh,\\
10 Crichton Street, Edinburgh EH8 9AB, UK
\and
Inria Paris, France}

\title{Quantum Differential and Linear Cryptanalysis}

\maketitle


\begin{abstract}
  Quantum computers, that may become available one day, would impact many
  scientific fields, most notably cryptography since many
  asymmetric primitives are insecure against an adversary with
  quantum capabilities.  Cryptographers are already anticipating this
  threat by proposing and studying a number of potentially
  \textit{quantum-safe} alternatives for those primitives.  On the other
  hand, symmetric primitives seem less vulnerable against quantum
  computing: the main known applicable result is Grover's algorithm that gives a
  quadratic speed-up for exhaustive search.

  In this work, we examine more closely the security of symmetric
  ciphers against quantum attacks.  Since our trust in symmetric ciphers
  relies mostly on their ability to resist cryptanalysis
  techniques, we investigate quantum cryptanalysis techniques.  More
  specifically, we consider quantum versions of differential and
  linear cryptanalysis.

  We show that it is usually possible to use quantum computations to obtain
  a quadratic speed-up for these attack techniques, but the situation
  must be nuanced: we don't get a quadratic speed-up for all variants of
  the attacks.  This allows us to demonstrate the following non-intuitive result: the best attack in the classical world does not
  necessarily lead to the best quantum one.  We give some
  examples of application on ciphers LAC and KLEIN.  We also discuss the
  important difference between an adversary that can only perform
  quantum computations, and an adversary that can also make quantum
  queries to a keyed primitive.
\end{abstract}

\section{Introduction}

Large quantum computers would have huge 
consequences in a number of scientific fields. Cryptography would certainly be 
dramatically impacted:
for instance, Shor's factoring algorithm~\cite{DBLP:journals/siamcomp/Shor97} makes asymmetric primitives such as RSA totally insecure in a post-quantum world. 
Even if quantum computers are unlikely to become widely available in the next couple of years, the
cryptographic community has decided to start worrying about this threat and to study its impact. One compelling reason for taking action is that even current pre-quantum long-term secrets are at risk as it seems feasible for a malicious organization to simply store all encrypted data until it has access to a quantum computer. This explains why post-quantum cryptosystems, based for instance on lattices or codes,
 have become a very hot topic in cryptology, and
researchers are now concentrating their efforts in order to provide
efficient alternatives that would resist quantum adversaries.

In this paper, we focus on symmetric cryptography, the other main branch
of cryptography.  Symmetric primitives also suffer from a reduced ideal 
security in the quantum world, but this security reduction turns out to be much less drastic than for many
asymmetric primitives.  So far, the main quantum attack on symmetric algorithms follows from
Grover's algorithm~\cite{DBLP:conf/stoc/Grover96} for searching an unsorted database of size $N$ in $O{(N^{1/2})}$ time.  It can be
applied to any generic exhaustive key search, but merely offers a
quadratic speed-up compared to a classical attack.  Therefore, the current
consensus is that key lengths should be doubled in order to offer the
same security against quantum algorithms.  This was one of the
motivations to require a version of AES with a 256-bit key, that appears in the initial recommendations of the European PQCRYPTO project \cite{augot2015initial}:
\begin{quote}
``Symmetric systems are usually not affected by Shor's algorithm, but they are affected by Grover's algorithm. Under Grover's attack, the best security a key of length $n$ can offer is $2^{n/2}$, so AES-128 offers only $2^{64}$ post-quantum security. PQCRYPTO recommends thoroughly analyzed ciphers with 256-bit keys to achieve $2^{128}$ post-quantum security.''
\end{quote}

Doubling the key length is a useful heuristic, but a more accurate
analysis is definitely called for.  Unfortunately, little work has been
done in this direction.  Only recently, a few results have started to
challenge the security of some symmetric cryptography constructions
against quantum adversaries.  In particular, some works have studied
generic attacks against symmetric constructions, or attacks against
modes of operations.

First, the quantum algorithm of Simon~\cite{simon1997power}, which is
based on the quantum Fourier transform, has been used to obtain a
quantum distinguisher for the 3-round Feistel cipher~\cite{5513654}, to
break the quantum version of the Even-Mansour scheme~\cite{6400943}, and
in the context of quantum related-key attacks \cite{roetteler2015note}.
More recently, the same quantum algorithm has been used to break widely
used block cipher modes of operations for MACs and authenticated
encryption~\cite{simoncrypto} (see also \cite{SS}).  All these attacks have a complexity
linear in the block size, and show that some constructions in
symmetric cryptography are badly broken if an adversary can make quantum
queries.

Kaplan~\cite{DBLP:journals/corr/Kaplan14} has also studied the quantum
complexity of generic meet-in-the-middle attacks for iterated block
ciphers constructions.  In particular, this work shows that having
access to quantum devices when attacking double iteration of block
ciphers can only reduce the time by an exponent 3/2, rather than the
expected quadratic improvement from Grover's algorithm. In consequence,
in stark contrast with classical adversaries, double iteration of block
ciphers can restore the security against quantum adversaries.

These are important steps in the right direction, providing the quantum
algorithms associated to some generic attacks on different
constructions.  These results also show that the situation is more
nuanced than a quadratic speed-up of all classical attacks.  Therefore,
in order to get a good understanding of the actual security of symmetric
cryptography constructions against quantum adversaries, we need to
develop and analyze quantum cryptanalytic techniques.  In particular, a
possible approach to devise new quantum attacks is to \emph{quantize}
classical ones.

{\bfseries\sffamily Security of symmetric key ciphers.}
While the security of crypto-systems in public key cryptography relies on the hardness of some well-understood mathematical problems, the security of symmetric key cryptography is
more heuristic.  Designers argue that a scheme is secure by proving its resistance against some particular attacks. This means that only cryptanalysis and security evaluations can bring confidence in a primitive.  Even when
a primitive has been largely studied, implemented and standardized, it
remains vital to carry on with the cryptanalysis effort using new
methods and techniques. 
Examples of standards that turned out to
be non-secure are indeed numerous (MD5, SHA1, RC4\ldots).
Symmetric security and confidence are therefore
exclusively based on this constant and challenging task of
cryptanalysis. 

Symmetric cryptanalysis relies on a toolbox of classical techniques such as
differential or linear cryptanalysis and their variants,
 algebraic attacks, etc. A cryptanalyst can study the
security of a cipher against those attacks, and evaluate the security
margin of a design using reduced-round versions.  This security margin
(how far the attack is from reaching all the rounds) is a good
measure of the security of a design; it can be used to compare
different designs and to detect whether a cipher is close to being broken.

Since the security of symmetric primitives relies so heavily on
cryptanalysis, it is crucial to evaluate how the availability of
quantum computing affects it, and whether dedicated attacks
can be more efficient than brute-force attacks based on Grover's algorithm.
In particular, we must design the toolbox of symmetric cryptanalysis
in a quantum setting in order to understand the security of symmetric
algorithms against quantum adversaries. 
In this paper, we consider quantum versions of cryptanalytic attacks for the first time\footnote{Previous results as~\cite{DBLP:journals/corr/Kaplan14,5513654,6400943} only consider quantizing generic attacks.},
evaluating how an adversary can perform some of the main attacks
on symmetric ciphers with a quantum computer.

{\bfseries\sffamily Modeling quantum adversaries.}  Following the
notions for PRF security in a quantum setting given by
Zhandry~\cite{DBLP:conf/focs/Zhandry12}, we consider two different
models for our analysis:
\begin{description}
\item[Standard security:] a block cipher is \emph{standard secure} against
quantum adversaries if no efficient quantum algorithm can distinguish the
block cipher from PRP (or a PRF) by making only \emph{classical}
queries (later denoted as $\Q1$).
\item[Quantum security:] a block cipher is \emph{quantum secure} against
quantum adversaries if no efficient quantum algorithm can distinguish the
block cipher from PRP (or a PRF) even by making \emph{quantum} queries
(later denoted as $\Q2$).
\end{description}

A $\Q1$ adversary collects data classically and processes them with quantum
operations, while a $\Q2$ adversary can directly query the cryptographic oracle
with a quantum superposition of classical inputs, and
receives the superposition of the corresponding outputs. 
The adversary, in the second model, is very powerful. Nevertheless, it is possible
to devise secure protocols against these attacks. In particular,
the model was used in~\cite{DBLP:conf/crypto/BonehZ13},
where quantum-secure signatures were introduced.
Later, the same authors showed how to construct message authentification codes
secure against~$\Q2$ adversaries~\cite{DBLP:conf/eurocrypt/BonehZ13}.
It was
also investigated in \cite{DBLP:conf/icits/DamgardFNS13} for secret-sharing schemes.
This model is also mathematically well defined, and it is 
convenient to use it to give security definitions against quantum adversaries,
a task that is often challenging~\cite{gagliardoni2015semantic}.
A more practical issue is that
even if the cryptographic oracle is designed to produce classical outcomes, 
its implementation may use some technology, for example optical fibers,
that a quantum adversary could exploit.
In practice, ensuring that only classical queries are allowed seems difficult,
especially in a world in which quantum resources become available.
It seems more promising to assume that security against quantum queries is not
granted and to study security in this model.

{\bfseries\sffamily Modes of operation.}
Block ciphers are typically used in a mode of operation, in order to
accommodate messages of variable length and to provide a specific
security property (confidentiality, integrity\ldots).  In classical
cryptography, we prove that modes of operations are secure, assuming
that the block cipher is secure, and we trust the block ciphers after
enough cryptanalysis has been performed.  We can do the same against
quantum adversaries, but proofs of security in the classical model do
not always translate to proofs of security in the quantum model.  In
particular, common MAC and AE modes secure in the classical model have
recently been broken with a Q2 attack~\cite{simoncrypto}.  On the other
hand, common encryption modes have been proven secure in the quantum
model~\cite{DBLP:conf/pqcrypto/AnandTTU16}, assuming either a
standard-secure PRF or a quantum-secure PRF.  In this work, we focus on the security of block
ciphers, but this analysis should be combined with an analysis of the
quantum security of modes of operation to get a full understanding of the
security of symmetric cryptography in the quantum model.

{\bfseries\sffamily Our results.}
We choose to focus here on differential cryptanalysis, the 
truncated differential variant, and on linear cryptanalysis.  We give for the first time a  synthetic description of
these attacks, and study how they are affected by the
availability of quantum computers.  As expected, we often get a quadratic
speed-up, but not for all attacks.

. 

In this work we use the concept of quantum walks to devise quantum attacks. This framework contains a lot of well known quantum algorithms such as Grover's search or Ambainis' algorithm for
element distinctness. More importantly, it allows one to compose these algorithms in the same way	 as classical algorithms can be composed. In order to keep our quantum attacks as simple as possible, we use a slightly modified Grover's search algorithm that can use quantum checking procedures. This simple trick comes at the cost of constant factors (ignored in our analysis), but a more involved approach, making better use of quantum walks may remove those additional factors.

We prove the following non-obvious results:
\begin{itemize}[noitemsep]
\item Differential cryptanalysis and linear cryptanalysis usually offer a \emph{quadratic gain} in
  the Q2 model over the classical model.
\item Truncated differential cryptanalysis, however, usually offers \emph{smaller
  gains} in the Q2 model.
\item Therefore, the optimal quantum attack is not always a quantum
  version of the optimal classical attack.
\item In the Q1 model, cryptanalytic attacks might offer \emph{little
  gain} over the classical model when the key-length is the same as the
  block length (\emph{e.g.}~AES-128).
\item But the gain of cryptanalytic attacks in the Q1
  model can be quite significant (\emph{similar to the Q2 model})
  when the key length is longer (\emph{e.g.}~AES-256).
\end{itemize}

The rest of the paper is organized as follows. We first present some
preliminaries on the classical (Section~\ref{sec:brute}) and quantum (Section~\ref{sec:algoquant}) settings. Section~\ref{sec:differential} treats differential attacks, while Section~\ref{sec:truncated} deals with truncated differential attacks and Section~\ref{sec:application} provides some applications on ciphers LAC and KLEIN. We study linear cryptanalysis in Section \ref{sec:linear}.
In Section~\ref{sec:discussion}, we discuss the obtained results. Section~\ref{sec:conclusion} concludes the paper and presents some open questions.

{\bfseries\sffamily Related work.}
Independently of our work, Zhou, Lu, Zhang and Sun~\cite{Zhou2015} introduced
quantum differential attacks. Their approach is, like ours, based on variations of
quantum search algorithms. In their case, these are quantum counting and quantum maximum finding.
Their results are, however, less general than ours. They only
consider attacks returning a subkey and not the whole key, and only in the Q2 model. Moreover,
the complexity analysis of our quantum attack is more fine-grained: it considers more parameters
which in principle leads to more efficient quantization of explicit classical differential attacks.

\section{Preliminaries}
\label{sec:brute}

In the following, we consider a block cipher $E$, with a blocksize of
$n$ bits, and a keysize of $k$ bits.  We assume that $E$ is an iterated
design with $r$ rounds, and we use $E^{(t)}$ to denote a
reduced version with $t$ rounds (so that $E = E^{(r)}$).
When the cipher $E$ is computed with a specific key $\kappa \in \{0,1\}^k$,
its action on a block $x$ is denoted by $E_\kappa(x)$.
The final goal of an attacker is to find the secret key $\kappa^*$ that was used to encrypt some
data. A query to the cryptographic oracle is denoted $E(x)$, where it is implicitly assumed
that $E$ encrypts with the key $\kappa^*$, \emph{i.e.}, $E(x)=E_{\kappa^*}(x)$.

{\bfseries\sffamily Key-recovery attack.}
The key can always be found using 
a brute-force attack; following our notations, the complexity of such a generic attack is
$2^{k}$. This defines the ideal security, \emph{i.e.}~the security a cipher should provide. Therefore, a cipher is considered {\em broken} if the key can be found ``faster'' than with the brute-force attack, where ``faster'' typically means with ``less encryptions''.
Three parameters define the efficiency of a specific attack.
The \emph{data complexity} is the number of calls to the cryptographic oracle $E(x)$.
The \emph{time complexity} is the time required to recover the key $\kappa^*$.
We consider that 
querying the cryptographic oracle requires one unit of time, so that the data complexity is included in the time complexity.
The \emph{memory complexity} is the memory needed to perform the attack.

{\bfseries\sffamily Distinguishers.}
Another type of attacks, less powerful than key-recovery ones, are distinguishers. Their aim is to distinguish a concrete cipher from an ideal one. A distinguishing attack often gives rise to a key-recovery attack and is always the sign of a weakness of the block cipher.

{\bfseries\sffamily Our scenario.}
In this paper, we consider some of the main families of non-generic attacks that can be a threat to some ciphers: differential and linear attacks. We propose their quantized version for the distinguisher and the last-rounds key-recovery variants of linear, simple differentials and truncated differentials. Our aim is to provide a solid first step towards ``quantizing'' symmetric families of attacks. To reach this objective, due to the technicality of the attacks themselves, and even more due to the technicality of combining them with quantum tools, we consider the most basic versions of the attacks.

{\bfseries\sffamily Success probability.} For the sake of simplicity, in this paper we do not take into account the success probability in the parameters of the attacks. In particular, because it affects in the same way both classical and quantum versions, it is not very useful for the comparison we want to perform. In practice, it would be enough to increase the data complexity by a constant factor to reach any pre-specified success probability. A detailed study of the success probability of statistical attacks can be found in~\cite{DBLP:journals/dcc/BlondeauGT11}.

\section{Quantum algorithms}
\label{sec:algoquant}
We use a number of quantum techniques in order to devise quantum attacks. 
Most of them are based
on well-known quantum algorithms that have been studied extensively over the last decades.
The equivalent to the classical brute-force attack in the quantum world is
to search through the key space using a Grover's search algorithm~\cite{DBLP:conf/stoc/Grover96},
leading to complexity
$2^{k/2}$.
Our goal is to devise quantum attacks that might be a threat to symmetric primitives by displaying a smaller complexity than the generic quantum exhaustive search.

\subsection{Variations on Grover's algorithm}

Although Grover's algorithm is usually presented as a search in an unstructured database, we use in our applications the following
slight generalization (see \cite{santha2008quantum} for a nice exposition on quantum-walk-based search algorithms).
The task is to find a marked element from a set $X$. We denote by $M \subseteq X$ the subset of marked elements and assume that we know a lower bound $\varepsilon$ on the fraction $|M|/|X|$ of marked elements. 
A classical algorithm to solve this problem is to repeat $O(1/\varepsilon)$ times: $(i)$ sample an element from $X$, $(ii)$ check if it is marked.

The cost of this algorithm can therefore be expressed as a function of two parameters: the \emph{Setup cost} \SSS, which is the cost of sampling a uniform element from $X$, and the \emph{Checking cost} \CC, which is the cost of checking if an element is marked.
The cost considered by the algorithm can be the time or the number of queries to the input. It suffices to consider specifically one of those resources when quantifying the Setup and Checking cost.

Similarly, Grover's algorithm~\cite{DBLP:conf/stoc/Grover96} is a quantum search procedure that finds a marked element, and whose complexity can be written as a function of the \emph{quantum Setup cost} \SSS, which is the cost of constructing a uniform superposition of all elements in $X$,
and the \emph{quantum Checking cost} \CC, which is the cost of applying a controlled-phase gate to the marked elements. 
Notice that a classical or a quantum algorithm that checks membership to $M$ can easily be modified to get a controlled-phase.

\begin{theorem}[Grover]
\label{thm:grover}
There exists a quantum algorithm which, with high probability, finds a marked element, if there is any, at cost of order $\frac{\SSS+\CC}{\sqrt{\varepsilon}}$.
\end{theorem}

In particular, the setup and the checking steps can themselves be quantum procedures. 
Assume for instance that 
the set $X$ is itself a subset of a larger set $\tilde X$. Grover's algorithm can then find an element $x\in X$ at a cost $({\tilde X / X})^{1/2}$, assuming that the setup and checking procedures are easy. Moreover, a closer look at the algorithm shows that if
one ignores the final measurement that returns one element, the algorithm produces a uniform superposition of the elements in $X$, which can be used to setup another Grover search.

Grover's algorithm can also be written as a special case of amplitude amplification,
a quantum technique introduced by Brassard, H\o yer and Tapp in order to boost the success probability of quantum algorithms~\cite{MR1947332}.
Intuitively, assume that a quantum algorithm $\mathcal A$ produces a 
superposition of outputs in a good subspace $G$ and outputs in a bad subspace $B$.
Then there exists a quantum algorithm that calls $\mathcal A$ as a subroutine
to amplify the amplitude of good outputs.

If $\mathcal A$ was a classical algorithm, repeating it $\Theta(1/a)$, where $a$ is the probability
of producing a good output, would lead to a new algorithm
with constant success probability. Just as Grover's algorithm, the amplitude amplification technique achieves
the same result with a quadratic improvement~\cite{MR1947332}. The intuitive reason is that quantum operations allow to amplify
the amplitudes of good output states, and that the corresponding probabilities are given by the squares of the amplitudes. Therefore, the amplification is quadratically faster than in the classical case.

\begin{theorem}[Amplitude amplification]
\label{algo:AA}
Let $\mathcal A$ be a quantum algorithm that, with no measurement, produces a superposition
$\sum_{x\in G} \alpha_x \ket x + \sum_{y\in B} \alpha_y \ket y$. 
Let $a=\sum_{x\in G} \vert \alpha_x \vert^2$ be the probability of obtaining, after measurement, a state in the \emph{good}
subspace $G$.

Then, there exists a quantum algorithm that calls $\mathcal A$ and $\mathcal A^{-1}$ as subroutines $\Theta(1/\sqrt{a})$ times and produces
an outcome $x\in G$ with a probability at least $\max(a,1-a)$.
\end{theorem}

A variant of quantum amplification amplitude can be used to count approximately, again with a quadratic speed-up over classical algorithms~\cite{DBLP:conf/icalp/BrassardHT98}. 

\begin{theorem}[Quantum counting]
\label{Qc}
Let $F : \{0, \ldots N-1\} \to \{0,1\}$ be a Boolean function, and $p = |F^{-1}(1)|/N$.
For every positive integer $D$, there is a quantum algorithm that makes $D$ queries to $F$ and, with probability at least $8/\pi^2$, outputs an estimate $p'$ to $p$ such that $|p-p'|\leq 2\pi\sqrt{p}/D+\pi^2/D^2$.
\end{theorem}

\subsection{Quantum search of pairs}
\label{sec:pairs}

We also use Ambainis' quantum algorithm for the element distinctness problem.
In our work, we use it to search for collisions.
\begin{theorem}[Ambainis~\cite{DBLP:journals/siamcomp/Ambainis07}]
\label{thm:ambainis}
Given a list of numbers $x_1,\ldots,x_n$, there exists a quantum algorithm that finds, with high probability, a pair of indices $(i,j)$
such that $x_i=x_j$, if there exists one, at a cost $O(n^{2/3})$.
\end{theorem}

The quantum algorithm proposed by Ambainis can easily be adapted
to finding a pair satisfying $x_i + x_j = w$ for any given $w$ (when the $x_i$'s are group elements and 
the ``$+$'' operation can be computed efficiently).

Ambainis' algorithm can also be
adapted to search  in a list $\{x_1, \ldots,x_n\}$ for a pair 
of indices $(i,j)$ such that $(x_i, x_j)$ satisfies some relation~$R$, with the promise that the
input contains at least $k$ possible pairs satisfying $R$.
If the input of the problem is a uniformly random set of pairs, it is sufficient, in order to find one, to
run Ambainis' algorithm on a smaller random subset of inputs.

\begin{theorem}
\label{thm:newalgo}
Consider a list of numbers $x_1,\ldots,x_n$ with $x_i \in X$ and a set of pairs
$\DD \subset X\times X$ such that $\DD$ contains exactly $k$ pairs.
There exists a quantum algorithm that finds, with high probability, a pair $(i,j)$
such that $(x_i,x_j) \in \DD$, at a cost $O(n^{2/3}k^{-1/3})$ on average over uniformly distributed inputs.
\end{theorem}

\begin{proof}
For a uniformly chosen subset $X' \subset X$ such that $|X'| = n/\sqrt{k}$, there is,
with constant probability, at least one
pair from $\DD$ in $X' \times X'$. According to Theorem~\ref{thm:ambainis}, 
the cost of finding this pair is $O(n^{2/3}k^{-1/3})$. 
Therefore, the quantum algorithm starts by sampling a random $X'$ and then runs Ambainis' algorithm on
this subset.

Notice that if the algorithm runs on uniformly random inputs, the set $X'$ does not need to be itself chosen at random. Any sufficiently large subset will contain one of the pairs with high probability, with high probability over the distribution of inputs.
\end{proof}

Before ending this section on quantum algorithms, we make a remark on 
the outputs produced by quantum-walk-based algorithms, such as
Ambainis' or Grover's algorithm.
In our applications, we use these not necessarily to produce some output, but to prepare a superposition of the outputs.
Similarly to Grover's algorithm, this can be done by running the algorithm without performing the
final measurement. 
However, since Ambainis' algorithm uses a quantum memory to maintain some
data structure, the superposition could in principle
include the data from the memory.
This issue does not happen with Grover's algorithm precisely because it does not require any data structure.

In our case, the algorithm ends in a superposition of nodes containing at least one of the searched pairs.
It has no consequence for our application, because we are nesting this procedure in Grover's algorithm.
Alternatively, it is possible to use amplitude amplification afterwards in order to amplify the
amplitude on the good nodes.
However, this could be an issue when nesting our algorithm in an arbitrary quantum algorithm.
For a discussion on nested quantum walks, see~\cite{nestedQW}.

\section{Differential Cryptanalysis}
\label{sec:differential}
Differential  cryptanalysis was introduced in~\cite{DBLP:conf/crypto/BihamS90} by Biham and Shamir.
It studies the propagation of differences in the input of a function ($\delta_{\mathrm{in}}$) and their influence on the generated output difference ($\delta_{\mathrm{out}}$).
In this section, we present the two main types of differential attacks on block ciphers in the classical world: the \emph{differential distinguisher} and the {\em last-rounds attack}, and then analyze their complexities for quantum adversaries. 

\begin{table}
  \centering
  \begin{tabular}{ll}
    \toprule
    $n$ & block-size\\
    $k$ & key-size\\
    $\Delta_{\mathin}$ & size ($\log$) of the set of input differences\\
    $\Delta_{\mathout}$ & size ($\log$) of the set of output differences\\
    $\Delta_{\mathfin}$ & size ($\log$) of the set of differences $\DD_\mathfin$ after last rounds\\
    $h_S$ & probability ($-\log$) of the differential characteristic ($h_S < n$)\\
    $h_T$ & probability ($-\log$) of the truncated differential characteristic\\
    $h_{\mathout}$ & probability ($-\log$) of generating $\delta_{\mathrm{out}}$ from $\DD_\mathfin$\\
    $k_{\mathout}$ & number of key bits required to invert the last rounds \\
    $C_{k_{\mathout}}$ & cost of recovering the last round subkey from a good pair\\
    $C_{k_{\mathout}}^{*}$ & quantum cost of recovering the last round subkey from a good pair\\
    $\varepsilon$ & bias of the linear approximation \\
    $\ell$ & number of linear approximations (Matsui's algorithm 1) \\
    \bottomrule
  \end{tabular}
  \caption{Notations used in the attacks.}
  \label{tab:notations}
\end{table}

\subsection{Classical Adversary}
 Differential attacks exploit the fact that there exists an input difference $\delta_{\mathrm{in}}$ and an output difference $\delta_{\mathrm{out}}$ to a cipher $E$ such that 
\begin{equation}\label{eq1}
h_S := - \log \Pr_x [E(x\oplus\delta_{\mathrm{in}}) = E(x) \oplus \delta_{\mathrm{out}}] < n,
\end{equation}
\emph{i.e.}, such that we can detect some non-random behaviour of the differences of plaintexts $x$ and $x \oplus\delta_{\mathrm{in}}$. Here, ``$\oplus$'' represents the bitwise xor of bit strings of equal length.
The value of $h_S$ is generally computed for a random key, and as usual in the literature, we will assume that 
Eq.~\eqref{eq1} approximately holds for the secret key $\kappa^*$.
Such a relation between $\delta_\mathin$ and $\delta_\mathout$ is typically found by studying the internal structure of the primitive in detail. While it seems plausible that a quantum computer could also be useful to find good pairs $(\delta_\mathin, \delta_\mathout)$, we will not investigate this problem here, but rather focus on attacks that can be mounted once a good pair satisfying Eq.~\eqref{eq1} is given.

\subsubsection{Differential Distinguisher} This non-random behaviour can already be used to attack a cryptosystem by distinguishing it from a random function. 
This distinguisher is based on the fact that, for a random function and a fixed $\delta_{\mathrm{in}}$, obtaining the $\delta_{\mathrm{out}}$ difference in the output would require $2^n$ trials, where $n$ is the size of the block. On the other hand, for the cipher $E$, if we collect $2^{h_S}$ input pairs verifying the input difference $\delta_{\mathrm{in}},$ we can expect to obtain one pair of outputs with output difference $\delta_{\mathrm{out}}$.
The complexity of such a distinguisher exploiting Eq.~\eqref{eq1} is 
$2^{h_S+1}$ in both data and time, and is negligible in terms of memory:
\begin{align}
T_{\mathrm{C}}^{\mathrm{s.\, dist.}} = D_{\mathrm{C}}^{\mathrm{s.\, dist.}} =  2^{h_S+1}.
\end{align}
Here, the subscript $\mathrm{C}$ refers to classical and $s.~dist.$ to ``simple distinguisher'' by opposition to its truncated version later in the text. 

Assuming that such a distinguisher exists for the first $R$ rounds of a cipher, we can transform the attack into a key recovery on more rounds by adding some rounds at the end or beginning of the cipher. This is called a {\em last-rounds attack}, and allows to attack more rounds than the distinguisher, typically one or two, depending on the cipher.

\subsubsection{Last-Rounds Attack} For simplicity and
without loss of generality, we consider that the rounds added to the
distinguisher are placed at the end. We attack a total of $r=R+r_{\mathrm{out}}$ rounds, where $R$ are the rounds covered by the distinguisher.
The main goal of the attack is to reduce the key space that needs to be searched exhaustively from $2^k$ to some $2^{k'}$ with $k' <k$. 
For this, we use the fact that we have an advantage for finding an input $x$ such that $E^{(R)}(x) \oplus E^{(R)}(x \oplus\delta_{\mathrm{in}}) = \delta_{\mathrm{out}}$. 

For a pair that generates the difference $\delta_{\mathrm{out}}$ after $R$ rounds, 
we denote by $\DD_\mathfin$ the set of possible differences generated in the output after the final $r_{\mathrm{out}}$ rounds,  the size of this set by $2^{\Delta_\mathfin} = |\DD_{\mathfin}|$. Let $2^{-h_{\mathrm{out}}}$ denote the probability of generating the difference $\delta_{\mathrm{out}}$ from a difference in $\DD_\mathfin$ when computing $r_{\mathrm{out}}$ rounds in the backward direction, and by $k_{\mathrm{out}}$ the number of key bits involved in these rounds.
The goal of the attack is to construct a list $L$ of candidates for the partial key that contains almost surely the correct value, and that has size strictly less than $2^{k_{\mathrm{out}}}$. 
For this, one starts with lists $L_M$ and $L_K$ where $L_M$ is a random subset of $2^{h_S}$ possible messages and $L_K$ contains all possible $k_{\mathrm{out}}$-bit strings. 
From Eq.~\eqref{eq1}, the list $L_M$ contains an element $x$ such that $E^{(R)}(x) \oplus E^{(R)}(x \oplus\delta_{\mathrm{in}}) = \delta_{\mathrm{out}}$ with high probability.
Let us apply two successive tests to the lists.

The first test keeps only the $x \in L_M$ such that $E(x) \oplus E(x \oplus\delta_{\mathrm{in}}) \in \DD_{\mathfin}$. The probability of satisfying this equation is $2^{\Delta_{\mathfin}-n}$. This gives a new list $L_M'$ of size $|L_M'| = 2^{h_S + \Delta_{\mathfin} -n}$. The cost of this first test is $2^{h_S+1}$.

The second test considers the set $L_M' \times L_K$ and keeps only the couples $(x,\kappa)$
such that $E^{(R)}_\kappa(x)+E^{(R)}_\kappa(x+\delta_{in})=\delta_{\mathrm{out}}$. This
is done by computing backward the possible partial keys for a given difference in $\DD_{\mathrm{out}}$.
Denote $C_{k_{\mathrm{out}}}$ the average cost of generating those keys for a given input pair.
Notice that $C_{k_{\mathrm{out}}}$ can be $1$ when the number of rounds added is reasonably small\footnote{For example, using precomputation tables with the values that allow the differential transitions through the S-Boxes.}, and is upper bounded by $2^{k_{\mathrm{out}}}$, that is, $1\leq C_{k_{\mathrm{out}}} \leq 2^{k_{\mathrm{out}}}$.
For a random pair $(x, \kappa)$, 
the probability of passing this test is $2^{-h_{\mathrm{out}}}$.
The size of the resulting set is therefore expected to be $2^{-h_{\mathrm{out}}} \times |L_M'| \times |L_K| = 2^{h_S + \Delta_{\mathfin} - n + k_{\mathrm{out}} - h_{\mathrm{out}}}$. 
The cost of this step is $C_{k_{\mathrm{out}}} 2^{h_S + \Delta_{\mathfin} - n}$.

The previous step produces a list of candidates for the partial key corresponding to the key bits
involved in the last $r_{\mathrm{out}}$ rounds and leading to a difference $\delta_{\mathrm{out}}$ after $R$ rounds. 
The last step of the attack consists in performing an exhaustive search 
within all partial keys of this set completed with all possible $k-k_{\mathrm{out}}$ bits.
The cost of this step is $2^{h_S + \Delta_{\mathfin} - n + k - h_{\mathrm{out}}}$.

In practice, the lists do not need to be built and everything can be performed ``on the fly''. Consequently, memory needs can be made negligible.
The total time complexity is:
\begin{align}
T_{\mathrm{C}}^{\mathrm{s.\, att.}} = 2^{h_S+1} +
			2^{h_S + \Delta_{\mathfin} - n} \left( C_{k_{\mathrm{out}}} +  2^{ k - h_{\mathrm{out}}} \right ),
\end{align} 
while the data complexity of this classical attack is $D_{\mathrm{C}}^{\mathrm{s.\, att.}} =  2^{h_S+1}$.
The attack is more efficient than an exhaustive search if $T_{\mathrm{C}}^{\mathrm{s.\, att.}} <2^k$.
 
\subsection{Quantum Adversary}

We first give attacks in the Q2 model, using superposition queries.

\subsubsection{Differential Distinguisher in the Q2 model}
The distinguisher consists in applying a Grover search over the set of messages $X$, of size $2^n$. 
More precisely, the algorithm makes $2^{h_S/2+1}$ queries to the encryption cipher, trying to find a marked
element $x \in M
= \left\{ x \in X \: : \: E(x \oplus \delta_{\mathrm{in}}) = E (x) \oplus \delta_{\mathrm{out}} \right\}.$
If it finds any, it outputs ``concrete''. If it does not, it outputs ``random''.

With the notations of the previous sections, the fraction of marked elements is $\varepsilon = 2^{-h_S}$ and Grover's algorithm finds a marked element after $\frac{1}{\sqrt{\varepsilon}} = 2^{h_S/2}$ iterations, each one requiring two queries to the encryption cipher. The time and data complexities are:
\begin{align}
T_{\Q2}^{\mathrm{s.\, dist.}} = D_{\Q2}^{\mathrm{s.\, dist.}} =  2^{h_S/2+1}.
\end{align}

It remains to prove that in the case of a random function, the probability of finding a marked element is negligible.
Assume that the probability of finding a marked element after $2^{h_S/2}$ quantum queries
is $\delta$. Then, this can be wrapped into an amplitude amplification procedure (Theorem~\ref{algo:AA}), leading
to a bounded error algorithm making $(1/\sqrt{\delta}) 2^{h_S/2}$ queries.
Since Grover's algorithm is optimal \cite{PhysRevA.60.2746}, we get that $(1/\sqrt{\delta}) 2^{h_S/2} \geq 2^{n/2}$, leading to $\delta\leq 2^{h_S-n}$.

\subsubsection{Last-Rounds Attack in the Q2 model}
\label{sec:lr-FQ}

An important point of the attack in the Q2 model is that it should avoid creating lists.
Instead, the algorithm queries the cryptographic algorithm whenever it needs
to sample an element from the list. 

The quantum attack can be described as a Grover search, with quantum procedures for the setup and checking phases.
The algorithm searches in the set $X= \{x: E(x\oplus \delta_{\mathrm{in}} ) \oplus E(x) \in \DD_{\mathfin}\}$ 
for a message such that $E^{(R)}(x\oplus \delta_{\mathrm{in}} ) \oplus E^{(R)}(x) = \delta_\mathout$.
This procedure outputs a message and when it is found, it suffices to execute the sequence corresponding to the checking of the Grover
search once more: generate partial key candidates and search among them, completed with all possible remaining $k-k_\mathout$ bits. This outputs the correct key and only adds a constant overhead factor to Grover search.
Notice that using tailor-made quantum walks, it should be possible to suppress this overhead.
Here we use Grover search to keep the attacks as simple as possible.

The setup phase prepares a uniform superposition of the $x\in X$; this
costs $\SSS=2^{(n-\Delta_{\mathfin})/ 2}$ using Grover's algorithm.
The checking phase takes a value $x$ and must determine whether $(x, x
\oplus \delta_{\mathrm{in}})$ is a good pair; it consists of the following successive steps:
\begin{enumerate}[noitemsep]
\item Compute all possible partial keys $\kappa_{\mathrm{out}}$ for the
  $k_{\mathrm{out}}$ bits that intervene in the last
  $r_\mathout$~rounds, assuming that $(x, x \oplus
  \delta_{\mathrm{in}})$ is a good pair ($E^{(R)}(x\oplus \delta_{\mathrm{in}} ) \oplus E^{(R)}(x) = \delta_\mathout$); 
  \item Complete the key by searching exhaustively using a Grover search, checking if the obtained key is the correct one.
\end{enumerate}
The cost of computing all possible partial keys is $C^*_{k_{\mathrm{out}}}$. The number of partial keys is
$2^{k_{\mathrm{out}}-h_{\mathrm{out}}}$, then completed by ${k-k_{\mathrm{out}}}$ remaining bits.
The cost of checking through all of them is thus $\CC=C^*_{k_{\mathrm{out}}}+2^{(k-h_{\mathrm{out}})/2}$.

The procedure succeeds whenever a
message $x$ is found such that $E^{(R)}(x) \oplus E^{(R)}(x \oplus\delta_{\mathrm{in}})=\delta_{\mathrm{out}}$. Therefore,
the probability of finding a marked element is lower bounded by
$\varepsilon \geq 2^{-h_S - \Delta_{\mathfin} + n}$. This is the conditional probability
of getting $E^{(R)}(x\oplus \delta_{\mathrm{in}} ) \oplus E^{(R)}(x) = \delta_\mathout$
given that the output
difference is in~$\DD_\mathfin$.

The total cost of the attack in the Q2 model is:
\begin{align}
T_{\Q2}^{\mathrm{s.\, att.}} =  2^{h_S/2+1}  +2^{(h_S + \Delta_\mathfin - n)/2} \left ( C^*_{k_{out}} + 2^{(k - h_\mathout)/2} \right), \quad D_{\Q2}^{\mathrm{s.\, att.}} =  2^{h_S/2+1},
\end{align}
with a data complexity identical to that of the distinguisher.

\subsubsection{Last-Rounds Attack in the Q1 model}
\label{sec:diff-lr-Q1}

We can also have a speed-up for the last-round attack in the Q1 model.
In this model, the quantum operations only take place after a classical
acquisition of the data.  In particular, the data complexity will be the
same as for a classical adversary.
After the first filtering step of the classical last-round attack, $2^{h_S - n +\Delta_{\mathfin}}$ couples satisfying $E(x) \oplus E(x\oplus \delta_\mathin) \in  \DD_\mathfin$ 
are obtained.  The attacker then uses a quantum algorithm to generate the partial keys $\kappa_\mathout$,
and a Grover search among those, completed  with all possible remaining $k - k_\mathout$ bits of the key,
in order to find the key.
This leads to data and time complexities of:
\begin{align}
\label{eq:timesemiq}
D_{\Q1}^{\mathrm{s.\, att.}} =  2^{h_S+1}, \quad T_{\Q1}^{\mathrm{s.\, att.}} =  2^{h_S+1}  +2^{(h_S + \Delta_\mathfin - n)/2} \left ( C^*_{k_{out}} + 2^{(k - h_\mathout)/2} \right),
\end{align}
where $C_{k_\mathout}^*$ denotes the average time complexity of
generating the partial keys of length $k_\mathout$, on a \emph{quantum}
computer.

Let us point out that any classical attack with data complexity smaller than square root of the exhaustive search of the key can be translated into an effective attack also in the Q1 model. This is more likely to happen for larger keys, where the limiting terms are often the second and third terms of $T_{\Q1}^{\mathrm{s.\, att.}}$ in Eq.~\eqref{eq:timesemiq}. See a detailed example in~\ref{klein}. The fact that long keys are more likely to ``maintain'' the validity of the attacks is an interesting result, as longer keys correspond to the recommendations for post-quantum symmetric primitives. In these cases, the Q1 model is particularly meaningful.

\subsubsection{Generating partial keys on a quantum computer}
\label{sect:Ckoutq}
We investigate further the average cost $C_{k_\mathout}$ of generating the partial keys compatible with some input~$x$ such that $E(x) \oplus E(x \oplus \delta_\mathin) \in \DD_\mathfin$. These partial keys correspond to the key bits involved in a transition from $\DD_\mathfin$ to $\delta_\mathout$.

Using a classical computer and a precomputation table, this usually takes constant time, but can be up to  $2^{k_\mathout}$ in the worst case. It turns out that the worst case this can be sped up using a quantum computer. 

First, when $k_\mathout < h_\mathout$, the search is expected to return
zero or one key candidates.  In this case, we run a Grover search over
the key, with complexity $C_{k_\mathout}^*=2^{k_{\mathrm{out}}/2}$.

Otherwise, let $K=2^{k_\mathout - h_\mathout}$ be the average number of
partial key candidates compatible with some input $x$, and denote
$N=2^{k_\mathout}$.  Finding one partial key can be done using Grover
search with $(N / K)^{1/2}$ steps. To find the second one, we modify the
checking procedure to exclude the first result, and
Grover's algorithm uses only $(N/(K-1))^{1/2}$ iterations.  In general,
to find the $i$-th candidate, we need $(N/(K-i))^{1/2}$ iterations, but
each iteration takes time proportional to $\sqrt{i-1}$, corresponding to the time needed to check that the candidate does not belong to the list $(i-1)$ excluded ones. This leads to the following upper bound on
$C_{k_\mathout}^*$, the quantum version of $C_{k_\mathout}$:
\begin{eqnarray*}
C_{k_\mathout}^* \leq  \sum_{i=0}^{K-1} \sqrt{\frac{N i}{K-i}} \approx K\sqrt{N} \int_0^1 \sqrt{\frac{x}{1-x}} d x = \frac{\pi}{2} K\sqrt{N}
\end{eqnarray*}
Neglecting the constant as usual, and replacing with our parameters, this gives $C_{k_\mathout}^* \lesssim 2^{3k_\mathout/2 - h_\mathout}$.

\section{Truncated Differential Cryptanalysis}
\label{sec:truncated}
Truncated differential cryptanalysis was introduced by Knudsen~\cite{DBLP:conf/fse/Knudsen94} in 94. Instead of fixed input and output differences, it considers sets of differences (like  the differences in the output in the last-rounds attack that we have considered in the previous section).

We assume in the following that we are given two sets
$\DD_{\mathin}$ and $\DD_\mathout$ of input and output differences such
that the probability of generating a difference in
$\mathcal{D}_{\mathrm{out}}$ from one in $\mathcal{D}_{\mathrm{in}}$ is
$2^{-h_T}$.  We further consider that $\DD_{\mathin}$ and
$\DD_{\mathout}$ are vector spaces.

\subsection{Classical Adversary}
As in the simple differential case, we first present the differential distinguisher based on the non-random property of the differences behaviour, and then discuss the last-rounds attack obtained from the truncated differential distinguishers.

\subsubsection{Truncated Differential Distinguisher}
Let $2^{\Delta_{\mathrm{in}}}$ and $2^{\Delta_{\mathrm{out}}}$ denote the sizes of the input and output sets of differences, respectively.
For simplicity and without loss of generality, we assume to have access to an encryption oracle, and therefore only consider the truncated differential as directed from input to output\footnote{In the case where the other direction provides better complexities, we could instead perform queries to a decryption oracle and change the roles of input and output in the attack. We assume that the most interesting direction has been chosen.}. 
We denote by $2^{-h_T}$ the probability of generating a difference in $\mathcal{D}_{\mathrm{out}}$ from one in $\mathcal{D}_{\mathrm{in}}.$
The condition for the distinguisher to work is that
$2^{-h_T}>2^{\Delta_{\mathrm{out}}-n}$. 
In this analysis, we assume that $2^{-h_T} \gg 2^{\Delta_{\mathrm{out}}-n}$.

The advantage of truncated differentials is that they allow the use of structures, \emph{i.e.}, sets of plaintext
values that can be combined into input pairs with a difference in $\mathcal{D}_{\mathrm{in}}$ in many different ways: one can generate $2^{2\Delta_{\mathrm{in}}-1}$ pairs using a single structure of size $2^{\Delta_{\mathrm{in}}}$. This reduces the data complexity compared to simple differential attacks. 

Two cases need to be considered.
If $\Delta_\mathin \geq (h_T+1)/2$, we build a single structure $\mathcal S$
of size $2^{(h_T+1)/2}$
such that for all pairs
$(x,y) \in \mathcal S \times \mathcal S$, $x\oplus y \in \DD_\mathin$.
This structure generates $2^{h_T}$ pairs.
If $\Delta_\mathin \leq (h_T+1)/2$, we have to consider multiple
  structures $\mathcal S_i$.  Each structure contains
  $2^{\Delta_\mathin}$ elements, and generates $2^{2\Delta_\mathin -1}$
  pairs of elements.  We consider $2^{h_T-2\Delta_\mathin +1}$ such
  structures in order to have $2^{h_T}$ candidate pairs.

In both cases, we have $2^{h_T}$ candidate pairs.  With high
probability, one of these pairs shall satisfy $E(x) \oplus E(y) \in
\DD_{\mathrm{out}}$, something that should not occur for a random function if $2^{-h_T} \gg
2^{\Delta_{\mathrm{out}}-n}$.  Therefore detecting a single valid pair gives an
efficient distinguisher.

The attack then works by checking if, for a pair generated by the data,
the output difference belongs to $\DD_\mathout$.  Since $\DD_\mathout$
is assumed to be a vector space, this can be reduced to trying to find a
collision on $n-\Delta_{\mathrm{out}}$ bits of the output. 
Once the data is generated, looking for a collision is not expensive
(\emph{e.g.} using a hash table), which means that time and data complexities coincide:
\begin{equation}\label{eq4}
D_{\mathrm{C}}^{\mathrm{tr.\, dist.}}=\max\{2^{(h_T+1)/2},2^{h_T-\Delta_{\mathrm{in}}+1}\}, \quad T_{\mathrm{C}}^{\mathrm{tr.\, dist.}}=\max\{2^{(h_T+1)/2},2^{h_T-\Delta_{\mathrm{in}}+1}\}.
\end{equation}

\subsubsection{Last-Rounds Attack}
\label{LR_T_C}
\label{sec:trunc-lr-C}

Last-rounds attacks work similarly as in the case of simple differential cryptanalysis. For simplicity, we assume that $r_{\mathrm{out}}$ rounds are added at the end of the truncated differential. 
The intermediate set of differences is denoted $\DD_\mathout$, and its size is $2^{\Delta_\mathout}$.
The set $\DD_\mathfin$, of size $2^{\Delta_\mathfin}$ denotes the possible differences for the outputs after the final round.
The probability of reaching a difference in $\DD_\mathout$ from a difference in $\DD_\mathin$ is $2^{-h_T}$, and
the probability of reaching a difference in $\DD_{\mathout}$ from a difference in $\DD_{\mathfin}$ is $2^{-h_{\mathrm{out}}}$.
Applying the same algorithm as in the simple differential case, the data complexity remains the same as for the distinguisher: 
\begin{align}
D_{\mathrm{C}}^{\mathrm{tr.\, att.}}=\max\{2^{(h_T+1)/2},2^{h_T-\Delta_{\mathrm{in}}+1}\}.
\end{align}
The time complexity in this case is:
\begin{align}\label{eq5}
T_{\mathrm{C}}^{\mathrm{tr.\, att.}}=\max\{2^{(h_T+1)/2},2^{h_T-\Delta_{\mathrm{in}}+1}\} +2^{h_T+\Delta_{\mathfin}-n} \left( C_{k_{\mathrm{out}}}+2^{k-h_{\mathrm{out}}} \right),
\end{align}
where $C_{k_{\mathrm{out}}}$ is the average cost of finding all the
partial key candidates corresponding to a pair of data with a difference
in $\DD_{\mathout}$. As mentioned earlier, $C_{k_{\mathrm{out}}}$ ranges from 1 
to $2^{k_{\mathrm{out}}}.$

\subsection{Quantum Adversary}

The truncated differential cryptanalysis is similar to the simple differential cryptanalysis, except that $\DD_{\mathrm{in}}$ and $\DD_{\mathrm{out}}$ are now sets instead of two fixed bit strings. 

\subsubsection{Truncated Differential Distinguisher}
Similarly to simple differential cryptanalysis, the distinguisher can only be more efficient in the Q2 model.
This comes from the fact that in both cases, the data complexity is the bottleneck. Since the Q1 model
does not provide any advantage over the classical one in data collection, there is no advantage in this model.

We use Ambainis' algorithm for element distinctness, given in Theorem~\ref{thm:ambainis},
in order to search for collisions inside the structures.
If a single structure is involved, 
the algorithm searches for a pair of  messages
$(x,y)$ in a set of size $2^{(h_T+1)/2}$,
such that $E(x) \oplus E(y) \in \DD_\mathout$.
Since there is, on average, only one such pair, this can be done using a quantum algorithm with 
$2^{(h_T+1)/3}$ queries.

If multiple structures are required, the strategy is to search
for one structure that contains a pair
$(x,y)$ such that $E(x) \oplus E(y) \in \DD_\mathout$.
This is done with a Grover search on the structure, 
using Ambainis' algorithm for the checking phase.
This returns a structure containing a desired pair, which is sufficient for the distinguisher.
The setup cost is constant.
The checking step, consisting in searching for a specific pair inside a structure of size $2^{\Delta_\mathin}$,
can be done with $\CC = 2^{2\Delta_{\mathrm{in}}/3}$ queries. Finally, since there is, with high probability,
at least one structure in $2^{h_T-2\Delta_{\mathrm{in}}+1}$ containing a pair such that  $E(x) \oplus E(y) \in \DD_\mathout$, we get a lower
bound on the success probability $\varepsilon \geq
2^{2\Delta_{\mathrm{in}}-h_T-1}$.
Using Theorem~\ref{thm:grover}, the total queries complexity is at most $2^{(h_T+1)/2 - \Delta_{\mathrm{in}} /3}$.

Combining both results leads to overall data and time complexities given by:
\begin{align}
D_{\Q2}^{\mathrm{tr.\, dist.}} = T_{\Q2}^{\mathrm{tr.\, dist.}} = \max \left\{ 2^{(h_T+1)/3}, 2^{(h_T+1)/2 - \Delta_{\mathrm{in}} /3} \right\}.
\end{align}

Similarly to the the quantum simple differential distinguisher, applying the same algorithm to a random function,
and stopping it after the same number of queries only provides a correct answer with negligible probability. 

\subsubsection{Last-Rounds Attack in the Q1 model}

As seen in Section~\ref{sec:trunc-lr-C}, last-round attacks for
truncated differential cryptanalysis are very similar to attacks with a
simple differential.  The attack in the Q1 model will differ from the
attack of Section~\ref{sec:diff-lr-Q1} only in the first
step, when querying the encryption function with the help of structures.
We start by generating a list of $2^{h_T}$ pairs with differences in $\mathcal D_\mathin$, which is done with data complexity:
\begin{align}
D_{\Q1}^{\mathrm{tr.\, att.}}  = \max\{2^{(h_T+1)/2},2^{h_T - \Delta_\mathin +1}\}.
\end{align}
The second step is to filter the list of elements to keep only the pairs $(x,y)$ such that $E(x) \oplus E(y) \in \DD_\mathfin$.
Notice that such a filtering can be done at no cost. It suffices to sort the elements  according to the values of their image, while constructing the list.

Finally, similarly to the Q1 simple differential attack, a quantum search algorithm is run on the filtered pairs, and the checking procedure consists in generating the partial key candidates
completed with $k-k_\mathrm{out}$ bits, and searching exhaustively for the key used in the cryptographic
oracle. 
In the Q1 model, the quantum speed-up only occurs in this step.

The average cost of generating  the partial keys on a quantum computer is denoted by $C^*_{k_\mathout}$. The average number of partial keys for a given pair of input is $2^{k_\mathout - h_\mathout}$. 
The fraction $\varepsilon$ of marked elements is $\varepsilon =2^{-h_T - \Delta_\mathfin + n}$, the setup cost is $S=1$ and the checking cost, a Grover search over the key space, is $\CC = C^*_{k_{out}} + 2^{(k - h_\mathout)/2}$. 
This gives a total cost:
\begin{align}
T_{\Q1}^{\mathrm{tr.\, att.}}  = \max \left\{2^{(h_T+1)/2},  2^{h_T - \Delta_\mathin +1}\right\} +  
2^{(h_T + \Delta_\mathfin - n)/2} \left ( C^*_{k_{out}} + 2^{(k - h_\mathout)/2} \right).
\end{align}

\subsubsection{Last-Rounds Attack in the Q2 model}
In the Q2 model, we want to avoid building classical lists. Instead, we query the cryptographic oracle
each time we need to sample a specific element.
This is challenging in the case of truncated differential because the use of structures made of lists is crucial.
The idea is to query the elements of the list on the fly.

\paragraph{Case where $h_T \leq 2\Delta_\mathin-1$.}
It is possible
to get $2^{h_T}$ pairs with differences in $\DD_\mathin$ with a single structure, $\mathcal S$,
of size $2^{(h_T+1)/2}$.
The attack runs a Grover search over 
$X=\{(x,y) \in \mathcal S \times \mathcal S : E(x)\oplus E(y) \in \mathcal \DD_\mathfin\}$.
The checking procedure is the same as for the quantum simple differential attack. For a given a pair of inputs, it generates all possible partial keys, and completes them to try to get the key used by cryptographic oracle.
This procedure returns a pair $(x,y)$.
The final step is to execute
the checking procedure in Grover search once more, suitably modified
to return the key given the pair $(x,y)$.

We analyze the setup cost of the attack. To prepare a superposition of the pairs in $X$,
we use the quantum search algorithm given in Theorem~\ref{thm:newalgo}.
This algorithm searches in a list for a pair of elements with a certain property,
considering there exist $k$ such pairs.
In our case, the list of elements is $\mathcal S$ of size $2^{(h_T+1)/2}$.
The total number of elements such that $E(x) \oplus E(y) \in \DD_\mathfin$ is therefore $2^{h_T-n+\Delta_\mathfin}$.
The algorithm of Theorem~\ref{thm:newalgo} prepares a superposition
of elements in $X$ in time $\SSS= 2^{(h_T+1)/3-(h_T-n+\Delta_\mathfin)/3} = 2^{(n-\Delta_\mathfin +1)/3}$.
The cost of the checking procedure is $\CC=C^*_{k_\mathout} + 2^{(k-h_\mathout)/2}$, as before.
The procedure is successful whenever a pair $(x,y)$ such that $E^{(R)}(x) \oplus E^{(R)}(y) \in \DD_\mathout$
is found. 
Given that the search is among pairs satisfying $x\oplus y \in \DD_\mathin$ and $E(x) \oplus E(y) \in \DD_\mathfin$, 
the probability for a pair to be good is $\varepsilon = 2^{-h_T -\Delta_\mathfin +n}$.
This gives a total running time:
\[
T_{\Q2}^{\mathrm{tr.\, att.}}   = 2^{h_T/2 - (n-\Delta_\mathfin)/6} + 2^{(h_T +\Delta_\mathfin -n)/2}\left( C^*_{k_\mathout} + 2^{(k-h_\mathout)/2}\right).
\]

\paragraph{Case where $h_T > 2\Delta_\mathin-1$.}
Multiple structures $\mathcal S_i$ of size
$2^{\Delta_{\mathin}}$ are now required, where~$i$ goes from 1 to $2^{h_T -
  2 \Delta_\mathin +1}$.  We use a Grover search over the structures, to
locate the structure containing a good pair (as previously, we then
repeat the attack with the right structure, but the complexity of this step is
negligible).  In order to test whether a structure contains a good pair,
the checking procedure uses the algorithm given in the previous section.
When using the
full structure with $2^{\Delta_{\mathrm{in}}}$ values, it has complexity:
\[
C = 2^{\Delta_{\mathrm{in}} - (n+1-\Delta_\mathfin)/6} + 2^{\Delta_{\mathrm{in}}-(n+1-\Delta_\mathfin)/2}\left( C^*_{k_\mathout} + 2^{(k-h_\mathout)/2}\right).
\]
We also have $S = 1$ and $\varepsilon = 2^{2\Delta_{\mathrm{in}}-1-h_T}$, which gives the following
complexity:
\begin{align*}
T_{\Q2}^{\mathrm{tr.\, att.}}  &= 2^{(h_T+1)/2 - (n+1-\Delta_{\mathrm{fin}})/6} + 2^{(h_T +\Delta_{\mathrm{fin}} -n)/2}\left( C^*_{k_\mathout} + 2^{(k-h_\mathout)/2}\right),
\end{align*}
which coincides with the complexity of the previous paragraph (when neglecting constants).

Note that for some parameters ($2 \Delta_{\mathrm{in}} - 1 - n +
\Delta_{\mathrm{fin}} < 0$) the expected size of $X$ is less than 1.  In
this case, the above procedure is not valid (it would run less than one
iteration of a Grover search).  Instead, the checking procedure should
run Ambainis' algorithm to find a pair with output in
$\mathcal{D}_{\mathrm{out}}$ (with complexity $2^{2\Delta_{\mathrm{in}}/3}$, because the parameters are chosen
assuming there is a solution), and use the usual key search algorithm
with complexity $C^*_{k_\mathout} + 2^{(k-h_\mathout)/2}$.  In this case
the complexity becomes:
\begin{align*}
T_{\Q2}^{\mathrm{tr.\, att.}}  &= 2^{(h_T+1)/2-\Delta_{\mathrm{in}}/3} +
2^{(h_T+1)/2-\Delta_{\mathrm{in}}}\left( C^*_{k_\mathout} + 2^{(k-h_\mathout)/2}\right).
\end{align*}
Note that $2^{(h_T+1)/2-\Delta_{\mathrm{in}}/3}$ is the complexity of
the truncated differential distinguisher in the \Q2 model.  Moreover,
with those parameters, 
$2^{(h_T+1)/2-\Delta_{\mathrm{in}}/3} > 2^{(h_T+1)/2 - (n+1-\Delta_{\mathrm{fin}})/6}$
and 
$2^{(h_T+1)/2-\Delta_{\mathrm{in}}} > 2^{(h_T +\Delta_{\mathrm{fin}}
  -n)/2}$.

Overall, summarizing both cases, we obtain
\begin{align*}
T_{\Q2}^{\mathrm{tr.\, att.}}  = & 2^{(h_T+1)/2} \max \big\{ 2^{ - \Delta_{\mathrm{in}/3}} , 2^{ - (n+1-\Delta_{\mathrm{fin}})/6} \big\} \\
& + \max  \big\{ 2^{(h_T +\Delta_{\mathrm{fin}} -n)/2}, 2^{(h_T+1)/2-\Delta_{\mathrm{in}}} \big\} \Big( C^*_{k_\mathout} + 2^{(k-h_\mathout)/2}\Big).
\end{align*}
Note that the speedup is always less than quadratic.

\section{Applications on existing ciphers}
\label{sec:application}

In this section we describe three examples of classical and quantum
differential attacks against block ciphers.  We have chosen examples
of real proposed ciphers where some of the best known attacks are simple
variants of differential cryptanalysis.  This allows us to illustrate the
important counter-intuitive points that we want to highlight, by
comparing the best classical attacks and the best quantum attacks.  We first consider the block cipher used in the authenticated encryption scheme LAC~\cite{lac}, and build for it a classical simple differential distinguisher and a more efficient classical truncated distinguisher. We quantize these attacks, and obtain that the quantum truncated distinguisher performs worse than a generic quantum exhaustive search. In the next application we consider the lightweight block cipher KLEIN~\cite{DBLP:conf/rfidsec/GongNL11}. Its 64-bit key version, KLEIN-64, has been recently broken~\cite{DBLP:conf/fse/LallemandN14} by a truncated differential last-rounds attack. When quantizing this attack, we show that it no longer works in the quantum world, and therefore KLEIN-64 is no longer broken. Finally, we consider KLEIN-96 and the best known attack~\cite{DBLP:conf/fse/LallemandN14} against this cipher. We show that its quantum variant still works in the post-quantum world, but only in the \Q2 model. 
These applications illustrate what we previously pointed out and believe
to be particularly meaningful: block ciphers with longer keys, following
the natural recommendation for resisting to generic quantum attacks, are
those for which the truncated attacks are more likely to still break the
cryptosystem in the postquantum world. Consequently, it is crucial to understand and compute the optimized quantum complexity of the different families of attacks, as we have started doing in this paper.
 
\subsection{Application 1: LAC}\label{lac}
We now show an example where a truncated differential attack is
more efficient than a simple differential attack using a classical
computer, but the opposite is true with a quantum computer.

We consider the reduced version of LBlock~\cite{DBLP:conf/acns/WuZ11} used in LAC~\cite{lac}.  According to~%
\cite{DBLP:conf/sacrypt/Leurent15}, the best known differential for the full
16 rounds has probability $2^{-61.5}$.  This yields a classical
distinguisher with complexity $2^{62.5}$ and a quantum distinguisher
with complexity $2^{31.75}$.
The corresponding truncated differential has the following
characteristics\footnote{We consider the truncated differential with
  $\mathcal{D}_{\mathrm{in}} =  \text{\tt 000000000000**0*}$ and
  $\mathcal{D}_{\mathrm{out}} = \text{\tt 0000***00000**00}$.
  If the input differential is non-zero on all active bytes, a pair
  follows the truncated differential when 14 sums of active bytes cancel
  out, and 3 sums of active bytes don't cancel out.  This gives a
  probability $(15/16)^{6} \cdot (1/15)^{14} \approx 2^{-55.3}$.}:
\begin{align*}
n &= 64 &
\Delta_{\mathin} &= 12 &
\Delta_{\mathout} &= 20 &
\tilde h_T &\approx 55.3
\end{align*}
We note that $\tilde h_T > n - \Delta_{\mathrm{out}}$, which is too large to
provide a working attack.  However, $\tilde h_T$ only
considers pairs following a given characteristic, and we expect
additional pairs to randomly give an output difference in
$\mathcal{D}_{\mathrm{out}}$.  Therefore, we estimate the probability
of the truncated differential as $2^{h_T} = 2^{-44} + 2^{-55.3}$.  In
order to check this hypothesis, we implemented a reduced version of
LAC with 3-bit APN S-Boxes, and verified that a bias can be
detected\footnote{The truncated path for the reduced version has a
  probability $2^{h_T} = 2^{-33} + 2^{-40.5}$.  We ran 32 experiments
  with $2^{31}$ structures of $2^{9}$ plaintexts each.  With a random
  function we expect about $2^{31} \cdot 2^9\cdot(2^9-1)/2 \cdot
  2^{-33} = 32704$ pairs satisfying the truncated differential, and
  about $32890$ with LAC.  The median number of pairs we found is
  33050 and it was larger than $32704$ is 31 out of 32 experiments.
  This agrees with our predictions.}.  In every structure, the
probability that a pair follows the truncated differential is $2^{23}
\cdot 2^{h_T} = 2^{-21} + 2^{-32.3}$, rather than $2^{-21}$ for a random
permutation.

As explained in Section~\ref{sec:algoquant} (Theorem~\ref{Qc}), this bias can be detected after
examining $2 \cdot 2^{-21} \cdot 2^{32.3 \cdot 2} = 2^{44.6}$
structures, \emph{i.e.}~$2^{56.6}$ plaintexts in a classical attack
(following \cite{DBLP:journals/dcc/BlondeauGT11}).  In a quantum
setting, we use quantum
counting~\cite{DBLP:conf/icalp/BrassardHT98,mosca1998quantum,MR1947332}
and examine $4 \pi \cdot 2^{-21/2} \cdot 2^{32.3} \approx 2^{25.4}$
structures, for a total cost of $2^{25.4} \cdot 2^{2/3 \cdot 12} =
2^{33.4}$.

To summarize, the best attack in the classical setting is a truncated
differential attack (with complexity $2^{60.9}$ rather than
$2^{62.5}$ for a simple differential attack), while the best attack in the quantum setting is a simple
differential attack (with complexity $2^{31.75}$ rather than
$2^{33.4}$ for a truncated differential attack).  Moreover, the quantum truncated differential attack is
actually less efficient than a generic attack using Grover's
algorithm.

\subsection{Application 2: KLEIN-64 and KLEIN-96}\label{klein}

\subsubsection{KLEIN-64}\label{klein64}

We consider exactly the attack from~\cite{DBLP:conf/fse/LallemandN14}. We omit here the details of the cipher and the truncated differential, but  provide the parameters needed to compute the complexity.

When taking into account the attack that provides the best time complexity, we have\footnote{For the attacks from~\cite{DBLP:conf/fse/LallemandN14} on KLEIN,  $h_T$ is always bigger than $n-\Delta_{\mathin}$, but the distinguisher from $\Delta_{in}$ to $\Delta_{out}$ still works exactly as described in Section~\ref{LR_T_C} because we compare with the probability of producing the truncated differential path and not just the truncated differential.}: $h_T=69.5,$ $\Delta_{\mathin}=16,$ $\Delta_{\mathfin}=32,$ $k=64$, $k_{\mathout}=32,$ $n=64$, $C_{k_{\mathout}}=2^{20}$ and $h_{\mathout}=45.$

In this case, we can recover the time and data complexities from the original result as\footnote{The slight difference with respect to \cite{DBLP:conf/fse/LallemandN14} is because here we have not taken into account the relative cost with respect to one encryption, for the sake of simplicity.} $D=2^{54.5}$ and $T=2^{54.5}+2^{57.5}+2^{56.5}= 2^{58.2}$, which is considerably faster than exhaustive search ($2^{64}$), therefore breaking the cipher. 

In the quantum scenario, the complexity of the generic exhaustive
search, which we use to measure the security, is $2^{32}$. The cipher is
considered broken if we can retrieve the key with smaller complexity.
When considering the Q2 or the Q1 case, the final term is accelerated by
square root.  On the other hand the second term has a square root in
$2^{h_T-n+\Delta_{\mathfin}}$, which is then multiplied by
$C_{k_{\mathout}}^*$. As shown in Section~\ref{sect:Ckoutq}, since $h_\mathout > k_\mathout$, we have $C_{k_{\mathout}}^* = 2^{k_{\mathout}/2}=2^{16}$ instead of
$2^{20}$. Consequently, the second term becomes $2^{34.75}$, thus the attack does not work.

 We have seen here an example of a primitive broken in the classical world, but remaining secure\footnote{We want to point out that notions “not-secure” (i.e. can be attacked in practice) and “broken” (i.e. can be attacked faster than brute-force), are not the same, though they are difficult to dissociate. } in the quantum one, for both models. 

\subsubsection{KLEIN-96}\label{klein96}

Here we consider the attack of type~III given in~\cite{DBLP:conf/fse/LallemandN14}, as it is the only one with data complexity lower than $2^{48},$ and therefore the only possible candidate for providing also an attack in the Q1 model.

The parameters of this classical attack are: $h_T=78,$ $\Delta_{\mathin}=32,$ $\Delta_\mathfin=32,$ $k_{\mathout}=48,$ $n=64$, $C_{k_{\mathout}}=2^{30}$ and $h_{\mathout}=52.$ We compute and obtain the same complexities as the original results in time and data: $D=2^{47}$ and $T=2^{47}+2^{46+30}+2^{90}.$
When quantizing this attack, we have to compare the complexities with $2^{96/2}=2^{48}.$

In the Q1 model, as the complexitites are going to be very tight, we consider the constants that take into account the cost of a computation with respect to one encryption~\footnote{For the sake of simplicity, we ignore this constants in the other cases, as they barely change the complexities.} (the unit of exhaustive search). Then, we obtain $2^{47}+2^{23+24}\cdot\frac{13}{14}+2^{45}\cdot\frac{1}{14}=2^{47.96},$  , which is lower that $2^{48}$ encryptions , so the
attack still works. The second term comes from $C_{k_{\mathout}}^* 2^{(h_T-n+\Delta_{\mathout})/2}$, where we can compute $C_{k_{\mathout}}^*$ as before, obtaining $2^{48/2}=2^{24}.$  The constants come from the fact that the computations related to the second term cover 13 rounds out of 14, as the last one is not computed but serves as matching point. The third term constant reflects that considering only one round out of the 14, many keys are already be discarded.

In the Q2 model, the first term is reduced to $2^{34.17}$ and becomes
negligible, with the final complexity at $2^{34.17}+2^{47}+2^{45} = 2^{47.3}$, showing that the cryptosystem is broken.

\section{Linear Cryptanalysis}
\label{sec:linear}

Linear cryptanalysis was discovered in 1992 by
Matsui~\cite{DBLP:conf/eurocrypt/MatsuiY92,DBLP:conf/eurocrypt/Matsui93}.
The idea of linear cryptanalysis is to approximate the round function with a linear
function, in order to find a linear approximation correlated to the non-linear encryption function~$E$.
 We describe the linear approximations using linear
masks; for instance, an approximation for one round is written as
$E^{(1)}(x)[\chi'] \approx x[\chi]$  where $\chi$ and $\chi'$ are linear masks for the
input and output, respectively, and  $x[\chi]= \bigoplus_{i:\chi_i =1}  x_i$. 
Here, ``$\approx$'' means that the probability that the two
values are equal is significantly larger than with a random permutation.

 The cryptanalyst has to build linear
approximations for each round, such that the output mask of a round is
equal to the input mask of the next round.  The piling-up lemma is then
used to evaluate the correlation of the approximation for the full
cipher. 
As for differential cryptanalysis, we assume here that the linear approximation is given and use it with a quantum computer to obtain either a distinguishing attack or a key recovery attack. 
In this section, we consider linear distinguishers and key recovery
attacks following from Matsui's work~\cite{DBLP:conf/eurocrypt/Matsui93}.

\subsection{Classical Adversary}

\subsubsection{Linear distinguisher}

In the following, $C$ denotes the ciphertext obtained when encrypting the plaintext $P$
with the key $K$.
We assume that we know a linear approximation with masks $(\chi_P,
\chi_C, \chi_K)$ and constant term $\chi_0 \in \{0,1\}$ satisfying
$\Pr \big[C[\chi_C] = P[\chi_P] \oplus K[\chi_K] \oplus \chi_0\big] = (1+\varepsilon)/2,$
with $\varepsilon \gg 2^{-n/2}$; or,
omitting the key dependency:
\begin{align*}
\Pr \big[C[\chi_C] = P[\chi_P] \big] &= (1 \pm \varepsilon)/2.
\end{align*}
An attacker can use this to distinguish
$E$ from a random permutation. 
The attack requires $D = A/\varepsilon^2$ known plaintexts $P_i$ and the corresponding ciphertexts $C_i$, where
$A$ is a small constant (\emph{e.g.} $A = 10$).  
The attacker computes the observed bias $\hat\varepsilon = | 2
\#\left\{i : C_i[\chi_C] = P_i[\chi_P] \right\}/D - 1 |$, and concludes that the data is random
if $\hat\varepsilon \leq \varepsilon/2$ and that it comes from $E$ otherwise.

If the data is generated by a
random permutation, then the expected value of $\hat\varepsilon$ is $0$, whereas, if it is
generated by $E$, the expected value of $\hat\varepsilon$ is $\varepsilon$.
We can compute the success probability of the attack assuming that the
values of $C_i[\chi_C] \oplus P_i[\chi_P]$ are identically distributed
Bernoulli random variables, with parameter $1/2$ or
$1/2 \pm \varepsilon$.  From Hoeffding's inequality, we get:
\begin{align*}
 \mathrm{Pr} \Big[\hat\varepsilon \geq \varepsilon/2\Big| \text{random permutation} \Big]
&\leq 2\exp\left(-2\frac{\varepsilon^2 }{4^2}D \right) \leq 2\exp\left(-\frac{A}{8} \right),\\
\mathrm{Pr} \Big[\hat\varepsilon \leq \varepsilon/2 \Big| \text{cipher} \ E \Big] 
&\leq \exp\left(-2\frac{\varepsilon^2 }{4^2}D \right) \leq \exp\left(-\frac{A}{8} \right);
\end{align*}
both error terms can also be made arbitrarily small by increasing $A$.

Overall, the complexity of the linear distinguisher is 
\begin{align}
D_{\mathrm{C}}^{\mathrm{lin.\, dist.}} = T_{\mathrm{C}}^{\mathrm{lin.\, dist.}} = 1/\varepsilon^2.
\end{align}
As explained in Section~\ref{sec:brute}, we do not take into account the factor $A$ that depends
on the success probability, and keep only the asymptotic term in the complexity.

\subsubsection{Key-recovery using an $r$-round approximation (Matsui's Algorithm 1)}

The linear distinguisher readily gives one key bit according to the sign of the
bias: if $K[\chi_K] = 0$, then we expect $\# \left\{i : C_i[\chi_C]
  = P_i[\chi_P] \oplus \chi_0 \right\}> D/2$.  The attack can be repeated
with different linear approximations in order to recover more key bits.
If we have $\ell$ independent linear approximations $(\chi_P^j, \chi_C^j, \chi_K^j, \chi_0^j)$ with bias at
least $\varepsilon$, the total complexity is:
\begin{align} 
D_{\mathrm{C}}^{\mathrm{Mat.}1} = 1/\varepsilon^2, \qquad T_{\mathrm{C}}^{\mathrm{Mat.}1} =  
\ell  /\varepsilon^2 + 2^{k-\ell}.
\end{align}

\subsubsection{Last-rounds attack (Matsui's Algorithm 2)}

Alternatively, linear cryptanalysis can be used in a last-rounds attack
that will often be more efficient.  Following the notations of the
previous sections, we consider a total of $R + r_{\mathrm{out}}$ rounds,
with an $R$-round linear distinguisher $(\chi_P, \chi_{C'})$ with bias $\varepsilon$,
and we use partial decryption for the last $r_{\mathrm{out}}$ rounds.

We denote by $k_{\mathrm{out}}$ the number of key bits necessary to compute
$C'[\chi_{C'}]$, where
$C' = E^{-r_{\mathrm{out}}}(C)$ from $C$.  
The attack proceeds as follows:
\begin{enumerate}[noitemsep]
\item Initialize a set of $2^{k_{\mathrm{out}}}$ counters $X_{k'}$ to zero, for each key candidate.
\item For each $(P,C)$ pair, and for every partial key guess $k'$,
  compute $C'$ from $C$ and $k'$, and increment $X_{k'}$ if $P[\chi_P]
  = C'[\chi_{C'}]$.
\item This gives $X_{k'} =  \# \big\{P,C:
  E_{k'}^{-r_{\mathrm{out}}}(C)[\chi_{C'}] = P[\chi_P]\big\}$.
\item Select the partial key $k'$ with the maximal absolute value of $X_{k'}$.
\end{enumerate}
This gives the following complexity:
\begin{align}
D_{\mathrm{C}}^{\mathrm{Mat.}2} &= 1/\varepsilon^2 &
T_{\mathrm{C}}^{\mathrm{Mat.}2} &= 2^{k_{\mathrm{out}}}/\varepsilon^2 + 2^{k-k_{\mathrm{out}}},
\end{align}
where, as before, we neglect constant factors.

We note that this algorithm can be improved using a distillation phase
where we count the number of occurrences of partial plaintexts and
ciphertexts, and an analysis phase using only these counters rather the
full data set.  In some specific cases, the analysis phase can be
improved by exploiting the Fast Fourier Transform
\cite{DBLP:conf/icisc/CollardSQ07}, but we will focus on the simpler case
described here.

\subsection{Quantum Adversary}

\subsubsection{Distinguisher in the Q2 model}
As in the previous sections, a speed-up for distinguishers is only observed for the Q2 model. 
The distinguisher is based on the quantum approximate counting algorithm of Theorem \ref{Qc}. As in the classical case, the goal is to distinguish between two Bernoulli distributions with parameter $1/2$ and $1/2 +\varepsilon$, respectively.

Using the quantum approximate counting algorithm, it is sufficient to make $O(1/\varepsilon)$ queries in order to achieve an $\varepsilon$-approximation. 
The data complexity of the quantum distinguisher is therefore, 
\begin{align}
D_{\Q2}^{\mathrm{lin.\, dist.}} = T_{\Q2}^{\mathrm{lin.\, dist.}} = 1/\varepsilon,
\end{align}
which constitutes a quadratic speed-up compared to the classical distinguisher.

\subsubsection{Key-recovery using an $r$-round approximation in the Q1 model}
Each linear relation allows the attacker to recover a bit of the key using $1/\varepsilon^2$ data, as the classical model. Once $\ell$ bits of the key have been recovered, one can apply Grover's algorithm to obtain the full key. 
For $\ell$ linear relations, the attack complexity is therefore:
\begin{align}
D_{\Q1}^{\mathrm{Mat.}1} &= \ell/\varepsilon^2 & 
T_{\Q1}^{\mathrm{Mat.}1} &= \ell/\varepsilon^2 + 2^{(k-\ell)/2}.
\end{align}

\subsubsection{Key-recovery using an $r$-round approximation in the Q2 model}
Each linear relation allows the attacker to recover a bit of the key using $1/\varepsilon$ data. 
If there are $\ell$ such relations, the attack complexity is:
\begin{align}
D_{\Q2}^{\mathrm{Mat.}1} &=\ell/\varepsilon &
T_{\Q2}^{\mathrm{Mat.}1} &= \ell/\varepsilon + 2^{(k-\ell)/2}.
\end{align}

Note that we do not \textit{a priori} obtain a quadratic improvement for the data complexity compared to the classical model. This is because the same data can be used many times in the classical model, whereas it is unclear whether something similar can be achieved using Grover's algorithm.

\subsubsection{Last-rounds attack in the Q1 model}

As usual for the Q1 model, one samples the same quantity of data as in the classical model and stores it in a quantum memory. 
Then the idea is to perform two successive instances of Grover's algorithm: the goal of the first one is to find a partial key of size $k_{\mathout}$ for which a bias $\varepsilon$ is detected for the first $R$ rounds: this has complexity $ 2^{k_{\mathrm{out}}/2}/\varepsilon$ with quantum counting; the second Grover aims at finding the rest of the key and has complexity $2^{(k-k_{\mathrm{out}})/2}$. 
Overall, the complexity of the attack is 
\begin{align}
D_{\Q1}^{\mathrm{Mat.}2} &= 1/\varepsilon^2 &
T_{\Q1}^{\mathrm{Mat.}2} &= 1/\varepsilon^2 + 2^{k_{\mathrm{out}}/2}/\varepsilon + 2^{(k-k_{\mathrm{out}})/2}.
\end{align}

\subsubsection{Last-rounds attack in the Q2 model.}
The strategy is similar, but the first step of the algorithm, \emph{i.e.}~finding the correct partial key, can be improved compared to the Q1 model. One uses a Grover search to obtain the partial key, and the checking step of Grover now consists of performing an approximate counting to detect the bias. 
Overall, the complexity of the attack is 
\begin{align}
D_{\Q2}^{\mathrm{Mat.}2} &= 2^{k_{\mathrm{out}}/2}/\varepsilon &
T_{\Q2}^{\mathrm{Mat.}2} &= 2^{k_{\mathrm{out}}/2}/\varepsilon + 2^{(k-k_{\mathrm{out}})/2}.
\end{align}

\section{Discussion}
\label{sec:discussion}

In this section, we first recall all the time complexities obtained
through the paper. The data complexities correspond to the first term of
each expression for the differential attacks. Next, we discuss how these
results affect the post-quantum security of symmetric ciphers with
respect to differential and linear attacks. 
As a remainder, notations are given in Table~\ref{tab:notations}.

\paragraph{Simple Differential Distinguishers:}
\begin{align*}
T_{\mathrm{C}}^{\mathrm{s.\, dist.}} &=  2^{h_S+1} &
T_{\Q2}^{\mathrm{s.\, dist.}} &= 2^{h_S/2+1} 
\end{align*}

\paragraph{Simple Differential Last-Rounds Attacks:}
\begin{alignat*}{2}
T_{\mathrm{C}}^{\mathrm{s.\, att.}} &=2^{h_S+1} &&+2^{h_S + \Delta_{\mathfin} - n} \Big(C_{k_{\mathrm{out}}}  +  2^{k - h_{\mathrm{out}}}\Big) \\
T_{\Q1}^{\mathrm{s.\, att.}} &=2^{h_S+1} &&+2^{(h_S + \Delta_\mathfin - n)/2} \Big( C^*_{k_{out}} + 2^{(k - h_\mathout)/2} \Big) \\
T_{\Q2}^{\mathrm{s.\, att.}} &= 2^{h_S/2+1} &&+2^{(h_S + \Delta_\mathfin - n)/2} \Big( C^*_{k_{out}} + 2^{(k - h_\mathout)/2} \Big) 
\end{alignat*}

\paragraph{Truncated Differential Distinguishers:}
\begin{align*}
T_{\mathrm{C}}^{\mathrm{tr.\, dist.}}&=\max\{2^{(h_T+1)/2},2^{h_T-\Delta_{\mathrm{in}}+1}\} &
T_{\Q2}^{\mathrm{tr.\, dist.}} &= \max \big\{ 2^{(h_T+1)/3}, 2^{(h_T+1)/2 - \Delta_{\mathrm{in}} /3} \big\} 
\end{align*}

\paragraph{Truncated Differential Last-Rounds Attacks:}
\begin{alignat*}{3}
T_{\mathrm{C}}^{\mathrm{tr.\, att.}}&=
\max\big\{2^{(h_T+1)/2},2^{h_T-\Delta_{\mathrm{in}}+1}\big\} &&+2^{h_T+\Delta_{\mathfin}-n} &&\Big( C_{k_{\mathrm{out}}}+2^{k-h_{\mathrm{out}}} \Big)\\
T_{\Q1}^{\mathrm{tr.\, att.}} &=
\max \big\{2^{(h_T+1)/2},  2^{h_T - \Delta_\mathin +1}\big\} &&+
2^{(h_T + \Delta_\mathfin - n)/2} && \Big ( C^*_{k_{out}} + 2^{(k - h_\mathout)/2} \Big) \\
T_{\Q2}^{\mathrm{tr.\, att.}}  &=  2^{(h_T+1)/2} \max \big\{ 2^{ - \Delta_{\mathrm{in}/3}} , 2^{ - (n+1-\Delta_{\mathrm{fin}})/6} \big\}  &&+ \max  \big\{ 2^{(h_T +\Delta_{\mathrm{fin}} -n)/2}, 2^{(h_T+1)/2-\Delta_{\mathrm{in}}} \big\}&& \Big( C^*_{k_\mathout} + 2^{(k-h_\mathout)/2}\Big)
\end{alignat*}

\paragraph{Linear Distinguishers:}
\begin{align*}
T_{\mathrm{C}}^{\mathrm{lin.\, dist.}} &=  1/\varepsilon^2 &
T_{\Q2}^{\mathrm{lin.\, dist.}} &= 1/\varepsilon 
\end{align*}

\paragraph{Linear Attacks:}
\begin{align*}
T_{\mathrm{C}}^{\mathrm{Mat.}1} &= \ell  /\varepsilon^2 + 2^{k-\ell} &
T_{\mathrm{C}}^{\mathrm{Mat.2}} &= 2^{k_{\mathrm{out}}}/\varepsilon^2 + 2^{k-k_{\mathrm{out}}} \\
T_{\Q1}^{\mathrm{Mat.}1} &= \ell/\varepsilon^2 + 2^{(k-\ell)/2} &
T_{\Q1}^{\mathrm{Mat.2}} &= 1/\varepsilon^2 + 2^{k_{\mathrm{out}}/2}/\varepsilon+ 2^{(k-k_{\mathrm{out}})/2}.\\
T_{\Q2}^{\mathrm{Mat.}1} &= \ell/\varepsilon + 2^{(k-\ell)/2} &
T_{\Q2}^{\mathrm{Mat.2}} &=2^{k_{\mathrm{out}}/2}/\varepsilon + 2^{(k-k_{\mathrm{out}})/2}
\end{align*}

The first observation we make is that the cost of a quantum differential
or linear attack
is at least the square root of the cost of the corresponding classical
attack.  In particular, if a block cipher is resistant to classical
differential and/or linear cryptanalysis (\emph{i.e.}~classical attacks
cost at least $2^k$), it is also resistant to the corresponding quantum
cryptanalysis (\emph{i.e.}~quantum differential and/or linear attacks
cost at least $2^{k/2}$). However, a quadratic speed-up is not always
possible with our techniques; in particular truncated attacks might be less
accelerated than simple differential ones.

\paragraph{Q1 model \emph{vs} Q2 model.}
We have studied quantum cryptanalysis with the notion of standard
security (Q1 model with only classical encryption queries) and quantum
security (Q2 model with quantum superposition queries).  As expected,
the Q2 model is stronger, and we often have a smaller quantum
acceleration in the Q1 model.  In particular, the data complexity of
attack in the Q1 model is the same as the data complexity of classical
attacks.
Still, there are important cases where quantum differential or linear
cryptanalysis can be more efficient than Grover's search in the Q1 model,
which shows that quantum cryptanalysis is also relevant in the more
realistic setting with only classical queries.

\paragraph{Quantum differential and linear attacks are more threatening to
  ciphers with larger key sizes.}
Though it seems counter-intuitive, the fact is that larger key sizes
also mean higher security claims to consider a cipher as secure.  In the
complexity figure given above, the terms that depend on the key size
(the right hand size terms) are likely to be the bottleneck for ciphers
with long keys with respect to the internal state size.  In all the
attacks studied here, this term is quadratically improved using quantum
computation, in both models.
Therefore, attacks against those ciphers will get the
most benefits from quantum computers.  We illustrated this effect in
Section~\ref{klein}, by studying KLEIN with two different key
sizes.

This effect is very strong in the Q1 model because most attacks have a
data complexity larger than $2^{n/2}$ (because $h_S >
n/2$, $h_T > n/2$, or $\varepsilon < 2^{-n/4}$).  If
the keysize is equal to $n$, this makes those attacks less efficient
than Grover's search, but they become interesting when $k$ is larger than
$n$.  In particular, with $k \ge 2n$, the data complexity is always
smaller than $2^{k/2}$.

This observation is particularly relevant because the recommended
strategy against quantum adversaries is to use longer keys~\cite{augot2015initial}.
We show that with this strategy, it is likely that classical
attacks that break the cryptosystem lead to quantum attacks that also
break it, even in the Q1 model where the adversary only makes classical
queries to the oracle.

\paragraph{The best attack might change from the classical to the
  quantum world.} 
Since truncated differential attacks use collision finding in the data
analysis step, they do not enjoy a quadratic improvement in the quantum setting.
Therefore, as we show in Section~\ref{lac}, a truncated
differential attack might be the best known attack in the classical
world, while the simple differential might become the best in the
quantum world. In particular, simply quantizing the best known attack does not
ensure obtaining the best possible attack in the post-quantum world,
which emphasizes the importance of studying quantum symmetric
cryptanalysis.

More strikingly, there are cases where differential attacks are more efficient than brute
force in the classical world, but quantum differential attacks are not
faster than Grover's algorithm, as we show in the example of
Section~\ref{klein64}.

\section{Conclusion and open questions}
\label{sec:conclusion}

Our work is an important step towards building a quantum symmetric cryptanalysis toolbox. Our results have corroborated our first intuition that symmetric cryptography does not seem ready for the post-quantum world. This not a direct conclusion from the paper, though indirectly the first logical approach for quantum symmetric cryptanalysis would be to quantize the best classical attack, and that would simplify the task. As we know for sure applications where the best attacks might change exist, cryptanalysis must be started anew.
The non-intuitive behaviors shown in our examples of applications help to illustrate the importance of understanding how symmetric attacks work in the quantum world, and therefore, of our results.
For building trust against quantum adversaries, this work should be extended, and other classical attacks should be investigated. Indeed, we have concluded that quantizing the best known classical differential attacks may not give the best quantum attack. This emphasizes the importance of studying and finding the best quantum attacks, including all known families of cryptanalysis. 

We have devised quantum attacks that break
classical cryptosystems faster than a quantum exhaustive search. 
However, the quantum-walk-based techniques used here can only lead to polynomial speed-ups, and the largest gap is quadratic, achieved by Grover's algorithm. Although this is significant,
it can not be interpreted as a collapse of cryptography against quantum adversaries similar to public-key cryptography based
on the hardness of factoring. However, we already mentioned that attacks based on the quantum Fourier transform,
which is at the core of Shor's algorithm for factoring and does not fall in the framework of quantum walks, have been found
for symmetric ciphers~\cite{5513654,6400943,roetteler2015note,simoncrypto}. 
\\

We end by mentioning a few open questions that we leave for future work.
In this work, we have studied quantum versions of differential and linear cryptanalysis. In each of these cases, we were either given a differential characteristics or a linear approximation to begin with, and used quantum algorithms to exploit them to perform a key recovery attack for instance. A natural question is whether quantum computers can also be useful to come up with good differential characteristics or linear approximations in the first place.

So far, we have only scratched the surface of linear cryptanalysis by quantizing the simplest versions of classical attacks, that is excluding more involved constructions using counters or the fast Fourier transform. Of course, since the quantum Fourier transform offers a significant speed-up compared to its classical counterpart, it makes sense to investigate whether it can be used to obtain more efficient quantum linear cryptanalysis. 

A major open question in the field of quantum cryptanalysis is certainly the choice of the right model of attack. In this work, we investigated two such models. The  Q2 model might appear rather extreme and perhaps even unrealistic since it is unclear why an attacker could access the cipher in superposition. But this model has the advantage of consistency. Also, a cipher secure in this model will remain secure in any setting. On the other hand, the Q1 model appears more realistic, but might be a little bit too simplistic. In particular, it seems important to better understand the interface between the classical register that stores the data that have been obtained by querying the cipher and the quantum register where they must be transferred in order to be further processed by the quantum computer.

\section*{Acknowledgements}
This work was supported by the Commission of the European Communities
through the Horizon 2020 program under project number 645622 PQCRYPTO.
MK acknowledges funding through grants 
ANR-12-PDOC-0022-01 and ESPRC EP/N003829/1.

\newcommand{\etalchar}[1]{$^{#1}$}

\end{document}